\documentclass[11pt]{article}
\usepackage[utf8]{inputenc}
\usepackage[left=1in, right=1in, bottom=1.15in, top=1.1in]{geometry}

\usepackage{amsmath,amssymb}
\usepackage{mathtools}
\usepackage[backref, colorlinks,citecolor=blue,linkcolor=magenta,bookmarks=true]{hyperref}
\usepackage[nameinlink]{cleveref}

\usepackage[vlined,ruled,linesnumbered]{algorithm2e} 
\SetKwRepeat{Do}{do}{while}

\usepackage{verbatim} 

\usepackage{amsthm} 
\newtheorem{theorem}{Theorem}

\newtheorem{lemma}{Lemma}
\newtheorem{fact}{Fact}

\newtheorem{definition}{Definition}

\newcommand{\calA}{\mathcal{A}}
\newcommand{\calB}{\mathcal{B}}

\newcommand{\calL}{\mathcal{L}}

\newcommand\sgn[0]{{\rm sgn}}
\def\Tarski{\textsc{Tarski}} \def\calL{\mathcal{L}}

\usepackage{tikz}
\usepackage{subcaption}

\title{Improved Upper Bounds for Finding Tarski Fixed Points}
\author{Xi Chen \\Columbia University \\ \texttt{{xichen@cs.columbia.edu}}
  \and Yuhao Li \\Columbia University\\ \texttt{yuhaoli@cs.columbia.edu}\vspace{0.3cm}}

\date{}

\begin{document}

\maketitle

\begin{abstract}
We study the query complexity of finding a Tarski fixed point
  over the $k$-dimensional grid $\{1,\ldots,n\}^k$.
Improving on the previous best upper bound of 
  $\smash{O(\log^{\lceil 2k/3\rceil} n)}$ \cite{fearnley2020fasterarxiv}, we give a new algorithm with 
  query complexity $\smash{O(\log^{\lceil (k+1)/2\rceil} n)}$.
This is based on a novel decomposition theorem about a weaker variant of 
  the Tarski fixed point problem, where the input consists of  a monotone function $f:[n]^k\rightarrow [n]^k$
  and a monotone sign function $b:[n]^k\rightarrow \{-1,0,1\}$ and the goal is to find an $x\in [n]^k$ 
  that satisfies \emph{either}
  $f(x)\preceq x$ and $b(x)\le 0$ \emph{or} $f(x)\succeq x$ and $b(x)\ge 0$.
  \end{abstract}
 
\thispagestyle{empty}
\newpage 
\setcounter{page}{1}

\section{Introduction}

In 1955, Tarski \cite{tarski1955lattice} proved that every monotone\footnote{We
  say $f$ is monotone if $f(a)\preceq f(b)$ whenever $a\preceq b$.} function $f:L\rightarrow L$
  over a complete lattice $(L,\preceq)$ has a fixed point, i.e. $x\in L$
  with $f(x)=x$.
Tarski's fixed point theorem has extensive applications in 
  game theory and economics,
  where it has been used to establish the 
  existence of important solution concepts 
  such as pure equilibria in supermodular games \cite{topkis1979equilibrium,topkis1998supermodularity,milgrom1990rationalizability}.
As a byproduct, search problems for these solution concepts naturally reduce to
  the problem of finding Tarski fixed points, which motivates 
  the investigation of its computational complexity.
More compelling motivations for studying Tarski fixed points came from a recent work of Etessami, Papadimitriou, Rubinstein and Yannakakis \cite{etessami2020tarski}, where they 
discovered 
  new connections of the Tarski fixed point problem 
  with well studied complexity classes such as PPAD and PLS, as well as reductions
  from Condon's (Shapley's) stochastic games \cite{condon1992complexity} to the Tarski fixed point problem.
However, our current understanding of the complexity of Tarski fixed points
  remains rather limited, whether it is about the query complexity of finding a Tarski fixed point or the white box version (i.e., when the function is given as a Boolean circuit) of the problem (e.g., whether the problem is complete in the class CLS~\cite{fearnley2021complexity,goos2022further} as the intersection of PPAD and PLS). This is in sharp contrast with 
  Brouwer's fixed point theorem \cite{papadimitriou1994complexity,chen2008matching,chen2009complexity},
  the other fixed point theorem that played a major role in economics.

In this paper we study the query complexity of finding a Tarski fixed point 
  in the complete lattice $([n]^k,\preceq)$ over the $k$-dimensional grid $[n]^k=\{1,\ldots,n\}^k$
  and equipped with the natural partial
  order over $\mathbb{Z}^k$, where $a\preceq b$ if $a_i\le b_i$ for all $i\in [k]$.
An algorithm under this model is given $n$ and $k$ and 
  has  query access to an unknown monotone function $f$ over $[n]^k$.
Each round the algorithm can send a query $x\in [n]^k$ to reveal $f(x)$
  and the goal is to find a fixed point of $f$ using as few queries as possible.
We will refer to this problem as $\textsc{Tarski}(n,k)$.
  
Back in 2011, Dang, Qi, and Ye \cite{dang2011computational} obtained 
  an $\smash{O(\log^k n)}$-query algorithm for $\textsc{Tarski}(n,k)$
  when $k$ is fixed. Their algorithm is based on a natural
  binary search strategy over coordinates.
No progress had been made on the problem until recently.
In \cite{etessami2020tarski}, Etessami et al. showed that (among other results) the upper bound $O(\log^2 n)$ for
  $\textsc{Tarski}(n,2)$ \cite{dang2011computational} 
  is indeed tight (even for randomized algorithms), which suggested that the algorithm of Dang et al.    
  might be optimal for all fixed $k$.
However, surprisingly, Fearnley, P\'alv\"olgyi and Savani~\cite{fearnley2020fasterarxiv}
  recently showed that the algorithm of \cite{dang2011computational}
  is not optimal by giving an
   algorithm for $\textsc{Tarski}(n,k)$ with $\smash{O(\log^{\lceil 2k/3\rceil}n)}$  queries.
 \vspace{0.3cm}

\noindent\textbf{Our contribution.}
Our main result is an improved upper bound 
  for the complexity of $\textsc{Tarski}$:

\begin{theorem}\label{main theorem}
For any fixed $k$, there is an $O \big(\log^{\lceil (k+1)/2\rceil} n\big )$-query algorithm for
  $\textsc{Tarski}(n,k)$.
\end{theorem}

Our algorithm is based on a new variant of the Tarski fixed point problem which we refer~to~as $\Tarski^*$.
It is inspired by
  the $O(\log^2 n)$-query algorithm of \cite{fearnley2020fasterarxiv}  for $\Tarski(n,3)$ (its inner algorithm
    in particular).
Our main contribution is a novel decomposition theorem for $\Tarski^*$, which
  leads to a more efficient recursive scheme for performing binary search on coordinates
  of the grid. 
We discuss the variant $\Tarski^*$ and its decomposition theorem next.

\subsection{Sketch of the Algorithm}\label{sec:sketch}

The algorithm of \cite{fearnley2020fasterarxiv}
  is obtained by combining an $O(\log^2 n)$-query algorithm for $\textsc{Tarski}(n,3)$ and a decomposition theorem.
Their algorithm for $\textsc{Tarski}(n,3)$ consists of an outer algorithm and an $O(\log n)$-query inner algorithm.
Given $f:[n]^3\rightarrow [n]^3$ as the input function,
  the outer algorithm
  starts by running the  inner algorithm to solve the following problem: 
\begin{flushleft}\begin{itemize}
\item Find a point $x\in [n]^3$ with $x_3= \lceil n/2\rceil$ such that $x$ is either prefixed ($f(x)\preceq x$)
  or postfixed ($x\preceq f(x)$).
Note that even though we focus on a layer of the grid (with $x_3=\lceil n/2\rceil$),
  the condition on $x$ being either prefixed or postfixed is about all three dimensions.
\end{itemize}\end{flushleft}
Once such a point $x$ is found, the outer algorithm can shrink the search space 
  significantly by only considering $\calL_{x,(n,n,n)}$ if $x$ is postfixed,
  or $\calL_{(1,1,1),x}$ if $x$ is prefixed, where we write $\calL_{a,b}$ to denote
  the grid with points $c:a\preceq c\preceq b$.
In both cases we obtain a grid $\calL_{a,b}$ such that $a\preceq b$, $a\preceq f(a)$ and $f(b)\preceq b$.
These conditions together guarantee that $f$ maps $\calL_{a,b}$ to itself and $f$ has a fixed point in $\calL_{a,b}$ (see \Cref{lemma: FixedPinab}).
Given that the side length of a dimension goes down by a factor of $2$ after each call to the inner algorithm,
  it takes no more than $O(\log n)$ calls to reduce the search space
  to a grid $\calL_{a,b}$ with $b_i-a_i\le 1$ and then a fixed point can be found by brute force.
The query complexity of the overall algorithm of \cite{fearnley2020fasterarxiv} for $\Tarski(n,3)$ is $O(\log^2 n)$.

After obtaining the $O(\log^2 n)$-query algorithm for $\Tarski(n,3)$,
  \cite{fearnley2020fasterarxiv} uses it to solve higher dimensional $\textsc{Tarski}$ 
  by proving a \emph{decomposition theorem}: if 
  $\textsc{Tarski}(n,a)$ can be solved in $q(n,a)$ queries and $\textsc{Tarski}(n,b)$ can be solved in $q(n,b)$ queries, then $\textsc{Tarski}(n,a+b)$ can be solved in $O(q(n,a)\cdot q(n,b))$ queries. 
Combined with the $O(\log^2 n)$-query algorithm for $\textsc{Tarski}(n,3)$, they obtain 
  an $O(\log^{\lceil 2k/3\rceil}n)$-query algorithm for $\textsc{Tarski}(n,k)$.

Our key idea is to develop a new decomposition theorem directly on
  the problem solved by the inner algorithm of \cite{fearnley2020fasterarxiv}, 
  and only apply the outer algorithm at the very end.
More formally we refer to the following problem as $\Tarski^*(n,k)$:\footnote{Note that our formal definition 
  in \Cref{sec:reducetotarski} will look different; the problems they capture are the same though.}  
\begin{flushleft}\begin{itemize}
\item Given a monotone function $f:[n]^{k+1}\rightarrow [n]^{k+1}$,
find a point $x$ with $x_{k+1}=\lceil n/2\rceil$ such that $x$ is either prefixed or postfixed.
As mentioned earlier, the condition on $x$ being either prefixed or postfixed is about all $k+1$ dimensions.
\end{itemize}\end{flushleft}  
Similar to the outer algorithm of \cite{fearnley2020fasterarxiv},
  any algorithm for $\Tarski^*(n,k)$ can be used as~a~subroutine to solve $\Tarski(n,k+1)$ with an $O(\log n)$-factor overhead
  (see \Cref{lemma: Tarski*-Tarski}).

The main technical contribution of this work is the proof of a new decomposition theorem for $\Tarski^*$: 
  if $\Tarski^*(n, a)$ can be solved in $q(n, a)$ queries and 
  $\Tarski^*(n, b)$ can be solved in $q(n, b)$ queries, then $\Tarski^*(n, a + b)$ can be solved in 
  $O(q(n, a)\cdot q(n, b))$ queries.
Now despite sharing the same statement\hspace{0.04cm}/\hspace{0.04cm}recursion,
  the proof of our decomposition theorem requires a number of new technical ingredients compared to that of \cite{fearnley2020fasterarxiv}.
This is mainly due to 
  the extra coordinate (i.e., coordinate $k+1$) that appears in $\Tarski^*$ but not in the original $\Tarski$. 
  
One obstacle is that the solution found
  by $\textsc{Tarski}^*$ appears to be too weak to directly prove the new decomposition theorem.
In particular, if one gets a postfixed point $x\preceq f(x)$ as a solution to $\Tarski^*(n,k)$, 
  both $x_{k+1}=f(x)_{k+1}$ or $x_{k+1}<f(x)_{k+1}$ could happen, and 
  this uncertainty would cause the proof strategy adopted by \cite{fearnley2020fasterarxiv} to fail. 
Instead we introduce~a stronger variant of $\textsc{Tarski}^*$ called $\textsc{RefinedTarski}^*$ (see \Cref{def: RefinedTarski*}) which
  poses further conditions on
  its solution regarding coordinate $k+1$. Given the same input, $\textsc{RefinedTarski}^*$ asks for two points $p^{\ell}\preceq p^r$ with $p^{\ell}_{k+1}=p^r_{k+1}=\lceil n/2\rceil$ such that 
  $p^\ell$ is postfixed in the first $k$ coordinates, $p^r$ is prefixed in the first $k$ coordinates,
  and  one of the following three conditions hold: \vspace{0.15cm}
\begin{enumerate}
\item $p^{\ell}_{k+1}<f(p^{\ell})_{k+1}$; \vspace{0.1 cm}
\item $\smash{p^r_{k+1}>f(p^r)_{k+1}}$; or\vspace{0.1 cm}
\item $f(p^{\ell})_{k+1}-p^{\ell}_{k+1}=f(p^r)_{k+1}-p^r_{k+1}=0$.\vspace{0.15cm}
\end{enumerate}
While $\textsc{RefinedTarski}^*$ looks much stronger than $\Tarski^*$, surprisingly we show in \Cref{lemma: refinedtarski*}
  that it can be solved by a small number of calls to $\textsc{Tarski}^*$. 
With $\textsc{RefinedTarski}^*$ as the bridge, we are able to prove the new decomposition theorem
  and obtain the improved bound for $\Tarski$.

\begin{figure}[!tbp]
\centering
\begin{tikzpicture}
\node at (3,4) {\Cref{main theorem}};
\node at (5.7,3) {with \Cref{lemma: Tarski*-Tarski} (Algorithm \ref{alg: solve Tarski})};
\node at (3,2) {\Cref{lemma: Tarski* complexity}};
\node at (4.6,1) {with \Cref{thm: logn for 2-d}};
\node at (3,0) {\Cref{thm: decomposition}:};
\node at (11,0.7) {$\textsc{Tarski}^*(n,a)$};
\node at (11,-0.7) {$\textsc{RefinedTarski}^*(n,b)$};
\node at (15,-0.7) {$\textsc{Tarski}^*(n,b)$};
\node at (6,0) {$\textsc{Tarski}^*(n,a+b)$};
\node at (13.5,0) {\Cref{lemma: refinedtarski*} (Algorithm \ref{alg: solve RefinedTarski*})};
\node at (9,0) {Algorithm \ref{alg: solve Tarski*}};

\draw[stealth-] (3,3.7)--(3,2.3);
\draw[stealth-] (3,1.7)--(3,0.3);
\draw[stealth-] (7.5,0.1)--(8.8,0.7);
\draw[stealth-] (7.5,-0.1)--(8.8,-0.7);
\draw[stealth-] (12.7,-0.7)--(13.9,-0.7);
\end{tikzpicture}

\caption{A Proof Sketch}
\label{Proof Sketch}
\end{figure}
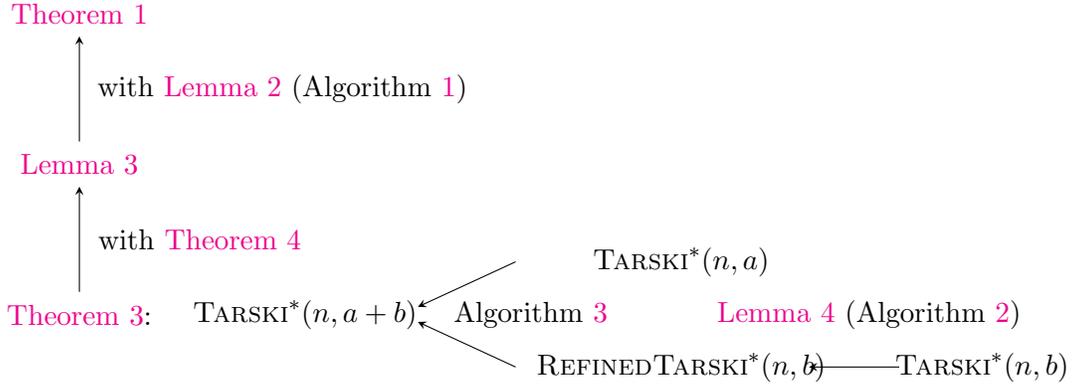

\def\LUB{\textsf{LUB}} \def\GLB{\textsf{GLB}}

\section{Preliminaries}

We start with the definition of \emph{monotone} functions and state Tarski's fixed point theorem:

\begin{definition}[Monotone functions]
	Let $(\calL,\preceq)$ be a complete lattice. A function $f: \calL\rightarrow \calL$ is said to be monotone
	if $f(a)\preceq f(b)$ for all $a,b\in \calL$ with $a\preceq b$.
\end{definition}

\begin{theorem}[Tarski]
	For any complete lattice $(\calL,\preceq)$ and any monotone function $f: \calL \rightarrow \calL$, there 
	  must be a point $x_0\in \mathcal{L}$ such that $f(x_0)=x_0$, i.e., $x_0$ is a fixed point.
\end{theorem}

In this paper we work on the query complexity of $\Tarski(n,k)$, i.e., the problem of finding a Tarski fixed point over a $k$-dimen\-sional grid $([n]^k,\preceq)$, where $[n]$ denotes $\{1,2,\cdots,n\}$ and $\preceq$ denotes the natural partial order over $\mathbb{Z}^k$: $a\preceq b$ if and only if $ a_i\leq b_i$ for every $i\in [k]$. For $a,b \in \mathbb{Z}^k$ with $a\preceq b$, we write $\calL_{a,b}$ to denote the set of points $x\in \mathbb{Z}^k$ with $a\preceq x\preceq b$.
A point $x\in[n]^k$ is called a \emph{prefixed} point of $f$ if $f(x)\preceq x$; a point $x\in[n]^k$ is called a \emph{postfixed} point of $f$ if  $x\preceq f(x)$.

	Let $S\subseteq \mathbb{Z}^k$ be a finite set of points. 
	A point $p\in \mathbb{Z}^k$ is an \emph{upper bound} of $S$ if $x\preceq p$ for all $x\in S$.
	We say $p$ is the \emph{least upper bound} of $S$ if $p$ is an upper bound of $S$ and $p\preceq q$ for 
	every upper bound $q$ of $S$ (i.e., $p_i=\max_{x\in S} x_i$ for all $i\in [k]$).
	Similarly, a point $p\in \mathbb{Z}^k$ is a lower bound of $S$ if $p\preceq x$ for all $x\in S$.
	We say $p$ is the \emph{greatest lower bound} of $S$ if $p$ is a lower bound of $S$ and $q\preceq p$ for 
	every lower bound $q$ of $S$ (i.e., $p_i=\min_{x\in S} x_i$ for all $i\in [k]$).
	We write $\LUB(S)$ and $\GLB(S)$ to denote the least upper bound and the greatest lower bound
	of $S$, respectively.

We record the following simple fact:

\begin{fact}\label{fact: lub glb}
Let finite $S,T\subseteq \mathbb{Z}^k$ be such that   $x\preceq y$ for all $x\in S,y\in T$.
Then $\emph{\LUB}(S)\preceq \emph{\GLB}(T)$. 
\end{fact}

We include a proof of the following simple lemma for completeness:

\begin{lemma}\label{lemma: FixedPinab}
Let $f: [n]^k\rightarrow [n]^k$ be a monotone function. 
Suppose $\ell,r\in [n]^k$ satisfy $\ell\preceq r$, $\ell\preceq f(\ell)$ and $f(r)\preceq r$.
Then $f$ maps $\mathcal{L}_{\ell,r}$ to itself and has a fixed point in $\calL_{\ell,r}$.
\end{lemma}
\begin{proof}
For any $x\in \calL_{\ell,r}$, we have from $\ell\preceq x\preceq r$ that
  $$\ell\preceq f(\ell)\preceq f(x)\preceq f(r)\preceq r$$ and thus, $f(x)\in \calL_{\ell,r}$. 
The existence of a fixed point in $\calL_{\ell,r}$ follows directly from Tarski's fixed point theorem
  applied on $f$ over $\calL_{\ell,r}$. 
\end{proof}

\section{Reduction to $\textsc{Tarski}^*$}\label{sec:reducetotarski}

For convenience we focus on $\Tarski(n,k+1)$.
Our algorithm for $\Tarski(n,k+1)$ (see Algorithm \ref{alg: solve Tarski}) over a monotone function $f:[n]^{k+1}\rightarrow [n]^{k+1}$ will first~set 
$$\ell=1^{k+1}:= (1,\ldots,1)\quad\text{and}\quad r=n^{k+1}:=(n,\ldots,n)$$ and then proceed to
  find a point $x\in [n]^{k+1}$ with $x_{k+1}=\lceil n/2\rceil$ that is
  either prefixed ($f(x)\preceq x$)  or postfixed ($x\preceq f(x)$).
Note that such a point $x$ must exist since by Tarski's fixed point theorem,
  there must be a point $x$ with $x_{k+1}=\lceil n/2\rceil$ such that $f(x)_i=x_i$ for all $i\in [k]$ (a fixed point over the slice $x_{k+1}=\lceil n/2\rceil$), and such
  a point must be either prefixed or postfixed; on the other hand, it is crucial that 
  the algorithm is not required to find an $x$ with $f(x)_i=x_i$ for all $i\in [k]$ but just an $x$
  that is either prefixed or postfixed.
After finding $x$, the algorithm replaces $r$ by $x$ if $x$ is prefixed, or $\ell$ by $x$ if $x$ is postfixed.
It follows from \Cref{lemma: FixedPinab} that $f$ remains a monotone function from $\calL_{\ell,r}$
  to itself but one of the $k+1$ dimensions gets cut by one half.
The algorithm recurses on $\calL_{\ell,r}$.

The key subproblem is to find such a point $x$ with $x_{k+1}=\lceil n/2\rceil$ that is either prefixed or postfixed,
  which we formulate as the following problem called $\Tarski^*(n,k)$: 

\begin{definition}[{\sc Tarski*$(n,k)$}]\label{def: Tarski*}
Given oracle access to a function $g: [n]^k\rightarrow \{-1,0,1\}^{k+1}$ satisfying
	\begin{itemize}
		\item For all $ x\in[n]^k$ and $i\in[k]$, we have $x_i+g(x)_i\in [n]$; and
		\item For all $x,y\in [n]^k$ with $x\preceq y$, we have $(x,0)+g(x)\preceq (y,0)+g(y)$,
	\end{itemize}
	find a point $x\in [n]^k$ such that 
	  either $g(x)_i \leq 0$ for all $i\in [k+1]$ or $g(x)_i\geq 0$ for all $i\in [k+1]$.
\end{definition}

To see the connection between $\Tarski^*(n,k)$ and 
  the subproblem described earlier,
  one can define $g:[n]^k\rightarrow \{-1,0,1\}^{k+1}$ using $f:[n]^{k+1}\rightarrow [n]^{k+1}$ by letting,
  for each $x\in [n]^k$, 
  $$g(x)_{k+1}=\sgn\big(f(x,\lceil n/2\rceil)_{k+1}-\lceil n/2\rceil\big)\quad \text{and}\quad
  g(x)_i=\sgn\big(f(x,\lceil n/2\rceil)_i-x_i\big)$$ for each $i\in [k]$.
On the one hand, it is easy to verify that $g$ satisfies both conditions in \Cref{def: Tarski*} when
  $f$ is monotone.
On the other hand, every $x\in [n]^k$ with  
  $\{-1,1\}\not\subseteq \bigcup_{i\in[k+1]}\{g(x)_{i}\}$ must satisfy that
   $(x,\lceil n/2\rceil)$ is either prefixed or postfixed in $f$. 

The next lemma shows how to use an algorithm for  $\textsc{Tarski*}(n,k)$ to solve $\textsc{Tarski}(n,k+1)$.
\begin{lemma}\label{lemma: Tarski*-Tarski}
	If $\textsc{Tarski*}(n,k)$ can be solved in $q(n,k)$ queries, then $\textsc{Tarski}(n,k+1)$ can be solved in $O(2^k+ k\log n\cdot q(n,k) )$ queries.
\end{lemma}
\begin{proof}
Suppose that $\calA$ is an algorithm for $\textsc{Tarski*}(n,k)$ with $q(n,k)$ queries. We present Algorithm \ref{alg: solve Tarski} and show that it can solve
  $\textsc{Tarski}(n,k+1)$ in $O(2^k+k\log n\cdot q(n,k))$ queries.

\textbf{Correctness.} The proof of correctness is based on the observation
  that $\ell\preceq r$, $\ell\preceq f(\ell)$ and $f(r)\preceq r$ at 
  the beginning of each while loop, which we prove below by induction. 
The basis is trivial. For the induction step, assume that 
  it holds at the beginning of the current while loop.
Then $f$ maps $\calL_{\ell,r}$ to itself and thus,
  $g$ satisfies both conditions in \Cref{def: Tarski*}.
As a result, $\calA$ can be used to find a point $q$ that is either prefixed
  or postfixed in $f$.
(Formally, one needs to embed $g$ over $\calL_{L,R}$ in
  the subgrid $\calL_{1^k,R-L}$ of $[n]^k$ and define $g':[n]^k\rightarrow \{-1,0,1\}^{k+1}$ such that solving $\Tarski^*(n,k)$ on $g'$ gives us $q$.)  
The way $\ell$ or $r$ is updated at the end of the loop
  makes sure that the statement holds at the beginning of the next loop.  
  
The last line of the algorithm makes sure that it returns a fixed point at the end.

\textbf{Query complexity.} Each while loop of Algorithm \ref{alg: solve Tarski}
  costs $q(n,k)$ queries.
	After each loop, 
	the side length of a dimension goes down by a factor of 2. So there are no more
	than $O(k \log n)$ rounds and thus, the query complexity of Algorithm \ref{alg: solve Tarski} is $O(2^k+ k\log n\cdot q(n,k))$.
\end{proof}

\begin{algorithm}[!t]
	\caption{Algorithm for $\textsc{Tarski}(n,k+1)$ via a reduction to $\Tarski^*(n,k)$}
	\label{alg: solve Tarski}
	\KwIn{Oracle access to a monotone function
		$f: [n]^{k+1}\rightarrow [n]^{k+1}$.
	}
	\KwOut{A fixed point $x\in [n]^{k+1}$ of $f$ with $f(x)=x$.
	}
	
	\BlankLine
	Let $\calA$ be an algorithm for $\textsc{Tarski*}(n,k)$. 
	Let $\ell=1^{k+1}$ and $r=n^{k+1}$. \label{solve tarski before while}\\
	
	\While{$|r-\ell|_\infty>2$}{
		Pick an $i\in [k+1]$ with $r_i-\ell_i>2$ and let 
		$$L=(\ell_1,\cdots,\ell_{i-1},\ell_{i+1},\cdots,\ell_{k+1}) 
		  \quad\text{and}\quad R=(r_1,\cdots,r_{i-1},r_{i+1},\cdots,r_{k+1}).$$\\
		Define a new function $g: \calL_{L ,R }\rightarrow \{-1,0,1\}^{k+1}$ as follows: 
		  $$g(x)\coloneqq \big(s_1,\cdots, s_{i-1}, s_{i+1},\cdots,s_{k+1},s_{i}\big)$$ 
		    where $s_j\coloneqq \sgn\left(f(x')_{j}-x'_{j}\right)$ and $x'=(x_1,\cdots,x_{i-1},\lceil (\ell_i+r_i)/2\rceil,x_{i},\cdots,x_{k})$.\\
		Run algorithm $\calA$ on $g$ to find a point $q\in \calL_{\ell,r}$
		  with $q_i= \lceil (\ell_i+r_i)/2\rceil$ that 
		  is either prefixed or postfixed;
		  set $r= q$ if $q$ is prefixed and 
		  set $\ell= q$ if $q$ is postfixed. 
	}
	Brute-force search $\calL_{\ell,r}$ to find a fixed point and return it.
\end{algorithm}

We prove the following upper bound for solving $\textsc{Tarski*}$ in the next section.
\begin{lemma}\label{lemma: Tarski* complexity}
	There is an $O(\log^{\lceil k/2 \rceil}n)$-query algorithm for $\textsc{Tarski*}(n,k)$.
\end{lemma}

\section{A Decomposition Theorem for $\textsc{Tarski*}$}
\label{sec: decomposition theorem}

The proof of \Cref{lemma: Tarski* complexity} 
   uses the following decomposition theorem for 
   $\textsc{Tarski*}$: 

\begin{theorem}\label{thm: decomposition}
	If $\textsc{Tarski*}(n,a)$ can be solved in $q(n,a)$ queries and $\textsc{Tarski*}(n,b)$ can be solved in $q(n,b)$ queries, then $\textsc{Tarski*}(n,a+b)$ can be solved in $O((b+1)\cdot q(n,a)\cdot q(n,b))$ queries.
\end{theorem}

We prove \Cref{thm: decomposition} in the rest of this section.
We also note that the algorithm of \cite{fearnley2020fasterarxiv} 
  can be used to solve the $2$-dimensional $\textsc{Tarski*}$ (see Theorem 14 in \cite{fearnley2020fasterarxiv}),
  even though they didn't define $\Tarski^*$ formally in the paper.
This leads to the following theorem about $\Tarski^*$:

\begin{theorem}[\cite{fearnley2020fasterarxiv}]
\label{thm: logn for 2-d}
	There is an $O(\log n)$-query algorithm for $\textsc{Tarski*}(n,2)$.
\end{theorem}
\Cref{lemma: Tarski* complexity}
 follows directly by combining \Cref{thm: decomposition} and \Cref{thm: logn for 2-d}. 


\subsection{A refined version of $\textsc{Tarski*}$}
We start the proof of our new decomposition theorem (\Cref{thm: decomposition}). 
To this end we first introduce a refined version of $\textsc{Tarski*}$.
\begin{definition}[$\textsc{RefinedTarski*}(n,k)$]\label{def: RefinedTarski*}
	Given a function $g:[n]^k \rightarrow \{-1,0,1\}^{k+1}$ satisfying
	\begin{itemize}
		\item For all $x\in[n]^k$ and $i\in[k]$, we have $x_i+g(x)_i\in [n]$; and
		\item For all $x,y\in [n]^k$ with $x\preceq y$, we have $(x,0)+g(x)\preceq (y,0)+g(y)$,
	\end{itemize}
	find a pair of points $p^\ell,p^r\in [n]^k$ such that $p^\ell\preceq p^r$,
	$$g(p^\ell)_{t}\ge 0\quad\text{and}\quad 
	g(p^r)_{t}\le 0,\quad \text{for all $t\in [k]$}$$ 
	and  one of the following conditions meets
	\begin{enumerate}
		\item $g(p^\ell)_{k+1}=1$;
		\item $g(p^r)_{k+1}=-1$; or
		\item $g(p^\ell)_{k+1}=g(p^r)_{k+1}=0$.
	\end{enumerate}
\end{definition}

We note that any solution $ p^\ell,p^r $ of $\textsc{RefinedTarski*}$ would imply a solution of corresponding $\textsc{Tarski*}$ problem directly by returning either $p^\ell$ or $p^r$. The following lemma shows that, in fact, 
   these two problems are computationally equivalent in their query complexity.

\begin{algorithm}[!t]
	\caption{Algorithm for $\textsc{RefinedTarski*}(n,k)$ via a reduction to $\Tarski^*(n,k)$}
	\label{alg: solve RefinedTarski*}
	\KwIn{Oracle access to $g: [n]^{k}\rightarrow \{-1,0,1\}^{k+1}$ 
	  that satisfies the conditions in \Cref{def: RefinedTarski*}.
	}
	\KwOut{
		A solution to $\textsc{RefinedTarski*}(n,k)$ on $g$.
	}
	
	\BlankLine
	Let $\calA$ be an algorithm for $\textsc{Tarski*}(n,k)$.
	Let $p^\ell= 1^k$ and $p^r= n^k$.\\
	Construct a new function $g^+: [n]^k\rightarrow \{-1,0,1\}^{k}\times\{-1,1\}$ as follows:
	\begin{equation*}
	\begin{cases} g^+(x)_i=g(x)_i,\quad \text{for all $i\in [k]$} \\[0.5ex]
	  \text{If $g(x)_{k+1}\geq 0$, then $g^+(x)_{k+1}=1$; if $g(x)_{k+1}=-1$, then $g^+(x)_{k+1}=-1$}\end{cases}
	\end{equation*}\\
	Run algorithm $\calA$ to find a solution $p^*$ to $\textsc{Tarski}^*(n,k)$ on $g^+$ over $[n]^k$.
	\label{callA 1}\\
	If $g^+(p^*)_{k+1}=1$, set $p^\ell\leftarrow p^*$; if $g^+(p^*)_{k+1}=-1$, set $p^r\leftarrow p^*$.\\
	If $g(p^*)_{k+1}\neq 0$, 
		\Return the pair of points $p^\ell,p^r$.\label{return}\\
	
	Construct a new function $g^-: [n]^k\rightarrow \{-1,0,1\}^{k}\times\{-1,1\}$ as follows:
	\begin{equation*}
	\begin{cases} 
	g^-(x)_i=g(x)_i,\quad \text{for all $i\in [k]$} \\[0.5ex]
	  \text{If $g(x)_{k+1}\leq 0$, then $g^-(x)_{k+1}=-1$; if $g(x)_{k+1}=1$, then $g^-(x)_{k+1}=1$}
	\end{cases}
	\end{equation*}\\
	Run algorithm $\calA$ to find a solution $q^*$ to $\textsc{Tarski}^*(n,k)$ on $g^-$ over $\calL_{p^{\ell},p^r}$. (This   can be done by embedding $g^-$ over $\calL_{p^{\ell},p^r}$ 
	inside $[n]^k$ and running $\calA$.)\label{callA 2}\\
	If $g^-(q^*)_{k+1}=1$, set $p^{\ell}\leftarrow q^*$; if $g^-(q^*)_{k+1}=-1$, set $p^r\leftarrow q^*$.\\
	\Return the pair of points $p^\ell,p^r$.
	
\end{algorithm}

\begin{lemma}\label{lemma: refinedtarski*}
	If $\textsc{Tarski*}(n,k)$ can be solved in $q(n,k)$ queries, then $\textsc{RefinedTarski*}(n,k)$ can be solved in $O(q(n,k))$ queries.
\end{lemma}
\begin{proof}
Suppose that $\calA$ is an algorithm to solve $\textsc{Tarski*}(n,k)$ with $q(n,k)$ queries. We show that 
    Algorithm \ref{alg: solve RefinedTarski*} will solve $\textsc{RefinedTarski*}$ in $O(q(n,k))$ queries.
	
\textbf{Correctness.}
	It is easy to verify that $g^+$ over $[n]^k$ satisfies both conditions of 
	\Cref{def: RefinedTarski*}. So
the point $p^*$ returned by algorithm $\calA$ on line \ref{callA 1} is either a prefixed or a postfixed point
  of $g^+$. If $g(p^*)_{k+1}=1$, then $p^\ell$ will be updated to $p^*$. By the definition of $\textsc{Tarski*}$ we have $g(p^*)_t\geq 0$ for all $t\in [k]$, which means $p^\ell,p^r$ will meet the first condition of $\textsc{RefinedTarski*}$. When $g(p^*)_{k+1}=-1$ $p^r$ will be updated to $p^*$ and $p^\ell,p^r$ will meet the second condition of $\textsc{RefinedTarski*}$.
	
	Now we can assume $g(p^*)_{k+1}=0$. So $p^\ell$ is updated to $p^*$ and $p^r$ remains $n^k$. Note that $p^*$ is a solution of $\textsc{Tarski*}$ under $g^+$ and $g^+(p^*)_{k+1}=1$ (because $g(p^*)_{k+1}=0$). So $g^+(p^*)_t\ge 0$ for all $t\in [k+1]$ and thus, $g(p^\ell)_t\geq 0$ for all $t\in [k]$. 
	
	Consider the point $q^*$ returned by algorithm $\calA$ on $g^-$ over $\calL_{p^\ell,p^r}$   on line \ref{callA 2}. If $g(q^*)_{k+1}=1$, then $p^\ell$ will be updated to $q^*$ and $p^l,p^r$ will meet the first condition of $\textsc{RefinedTarski*}$. Otherwise, we have $g(q^*)_{k+1}\leq 0$. Since $p^\ell\preceq q^*$, by the second property of function $g$, we know $0=g(p^\ell)_{k+1}\leq g(q^*)_{k+1}\leq 0$, i.e., $g(q^*)_{k+1}=0$. With the definition of $g^-$, we know that $g^-(q^*)=-1$. Note that $q^*$ is a solution to $\textsc{Tarski*}$ on $g^-$, so we have 
	  $g^-(q^*)_t\le 0$ for all $t\in [k+1]$. In this case, $p^r$ will be updated as $q^*$, so $p^\ell,p^r$ will meet the third condition of $\textsc{RefinedTarski*}$.

	\textbf{Query Complexity.} Algorithm \ref{alg: solve RefinedTarski*} just calls the algorithm $\calA$ at most two times on line \ref{callA 1} and line \ref{callA 2}, so the query complexity of Algorithm \ref{alg: solve RefinedTarski*} is $O(q(n,k))$.
\end{proof}

Now we are ready to prove \Cref{thm: decomposition}.

\begin{algorithm}[!t]
	\caption{Algorithm for $\textsc{Tarski*}(n,{a+b} )$ via $\Tarski^*(n,a)$ and 
	  $\Tarski^*(n,b)$}
	\label{alg: solve Tarski*}
	\KwIn{Oracle access to 
		$g: [n]^{a+b}\rightarrow \{-1,0,1\}^{a+b+1}$ satisfying conditions in \Cref{def: Tarski*}.
	}
	\KwOut{A solution to  $\textsc{Tarski*}(n,a+b)$ on $g$.
	}
	
	\BlankLine
	
	Let $\calA$ be an algorithm for $\textsc{Tarski*}(n,a)$  and $\calB$ be an algorithm for $\textsc{Tarski*}(n,b)$.\\
	Let $i \leftarrow 1$ be the round number. \\
	
	\Do{}{
	For each previous round $k\in[i-1]$, let $q^k\in [n]^b$ be the point queried by $\calB$ and $r^k\in \{-1,0,1\}^{b+1}$ be the answer.\\
		
		Given the sequence $((q^1,r^1),\cdots,(q^{i-1},r^{i-1}))$, let 
		  $q^i\in [n]^{b}$ be the $i$th query of $\calB$.\\

		Set (when $i=1$, set $p^{(\ell,1)}=1^a$ and $p^{(r,1)}=n^a$)
		\begin{align*}p^{(\ell,i)}&\leftarrow \LUB\left(\left\{p^{(\ell,k)}:k\in [i-1] \ \text{and}\ q^k\preceq q^i\right\}\right)\quad\text{and}\\[0.5ex]
		p^{(r,i)}&\leftarrow \GLB\left(\left\{p^{(r,k)}:k\in [i-1] \ \text{and}\  q^i\preceq q^k\right\}\right)
		\end{align*}\label{define pr}\\
		\For{each $j$ from $a+1$ to $a+b+1$}{\label{for loop}
			Define a new function $g_{j}: [n]^{a} \rightarrow \{-1,0,1\}^{a+1}$ as follows:
			$$ g_{j}(x)=\Big(g(x,q^i)_1,\cdots,g(x,q^i)_a,g(x,q^i)_j\Big),
			\quad \text{for every $x\in [n]^{a}$.}$$\\
			Run Algorithm \ref{alg: solve RefinedTarski*} with $\calA$ 
			  to find a solution 
			  $p^{(\ell,*)},p^{(r,*)}$ to $\textsc{RefinedTarski}^*(n,a)$
			  on $g_j$ over $\smash{\calL_{p^{(\ell,i)},p^{(r,i)}}}$
			  \label{find pl pr};
			set $p^{(\ell,i)}\leftarrow p^{(\ell,*)}$ and  $p^{(r,i)}\leftarrow p^{(r,*)}$.
			}

		Construct $r^i\in \{-1,0,1\}^{b+1}$ as the query result to $q^i$: $$r^i_{t-a}=g(p^{(\ell,i)},q^i)_t,\quad\text{for each $t\in[a+b+1]\setminus [a]$}.$$\\ 
		
		If $r_t^i\ge 0$ for all $t\in [b+1]$, 
		  \Return $(p^{(\ell,i)},q^i)$.\\
		
		If $r^i_t\le 0$ for all $t\in [b+1]$, 
		\Return $(p^{(r,i)},q^i)$.
	}
	
\end{algorithm}

\begin{proof}[Proof of \Cref{thm: decomposition}]
	Suppose $\calA$ is a query algorithm to solve $\textsc{Tarski*}(n,a)$ in $q(n,a)$ queries and $\calB$ is a query algorithm to solve $\textsc{Tarski*}(n,b)$ in $q(n,b)$ queries. We will show that Algorithm \ref{alg: solve Tarski*} can solve $\textsc{Tarski*}(n,a+b)$ in $O((b+1)\cdot q(n,a)\cdot q(n,b))$ queries. 
	
\def\hhh{h}	
	
\textbf{Overview.}
At a high level, Algorithm \ref{alg: solve Tarski*}   
  will run $\calB$ for $\Tarski^*(n,b)$ on a function $\hhh: [n]^{b}\rightarrow
  \{-1,0,1\}^{b+1}$ built on the go using the input $g:[n]^{a+b}
  \rightarrow \{-1,0,1\}^{a+b+1}$ of $\Tarski^*(n,a+b)$ (that satisfies the conditions in 
  \Cref{def: Tarski*}).
Let $q^1,\ldots,q^{i-1}\in [n]^b$ be the $i-1$ queries that 
  $\calB$ has made so far, for some $i\ge 1$, and 
  let $r^1,\ldots,r^{i-1}\in \{-1,0,1\}^{b+1}$ be the query results
  on $\hhh$.
Let $q^i\in [n]^b$ be the new query made by $\calB$ in the $i$th round.
Our challenge is to use $g$ (its restriction on points with the last
  $b$ coordinates being $q^i$) to come up with an $r^i\in \{-1,0,1\}^{b+1}$ 
  as the answer $h(q^i)$ to the query
  such that
\begin{flushleft}\begin{enumerate} \item
\Cref{lemma: rmonotonicity}: All results $(q^1,r^1),\ldots,(q^i,r^i)$ are consistent with the conditions of \Cref{def: Tarski*}, i.e., 
  $q^j_i+r^j_i\in [n]$ for all $i\in [b]$ and  
  $(q^j,0)+r^j\preceq (q^{j'},0)+r^{j'}$ for all $j,j'$ with $q^j\preceq q^{j'}$; and 
  \item \Cref{lemma: rsolution}: When $q^i$ is a solution to $\Tarski^*(n,b+1)$ on $h$,
  i.e., either $r^i_t\ge 0$ for all $t\in [b+1]$ or $r^i_t\le 0$ for all $t\in [b+1]$,
  we can use $q^i$ to obtain a solution to $\Tarski^*(n,a+b)$ on 
  the original input function $g$.
\end{enumerate}\end{flushleft}

To obtain $r^i$, we need to run $\calA$ $b+1$ times
  to obtain a pair of points $p^{(\ell,i)},p^{(r,i)}\in [n]^a$ and 
  use $g(p^{(\ell,i)},q^i)$ and $g(p^{(r,i)},q^i)$ to determine $r^i$.
A crucial component in the computation of $p^{(\ell,i)}$ and $p^{(r,i)}$
  is to initialize the search space using 
  pairs $p^{(\ell,j)},p^{(r,j)}$, $j\in [i-1]$, from previous rounds.

	\textbf{Correctness.} 
	We prove a sequence of lemmas about Algorithm \ref{alg: solve Tarski*}:

\begin{lemma}\label{hehehehe}
At the end of each round $i$, 
we have $p^{(\ell,i)}\preceq p^{(r,i)}$ and
\begin{equation}
\label{induc}
g(p^{(\ell,i)},q^i)_t\ge 0\quad\text{and}\quad
g(p^{(r,i)},q^i)_t\le 0,\quad \text{for all $t\in [a]$}.
\end{equation}
\end{lemma}
	\begin{proof}
	We start with the base case for the first round.
    We have $p^{(\ell,1)}=1^{a}\preceq n^{a}=p^{(r,1)}$ at the beginning
      so \Cref{induc} holds. 
    It is easy to prove by induction that both $p^{(\ell,1)}\preceq p^{(r,1)}$ and \Cref{induc} hold
      at the beginning of each for loop on
      line \ref{for loop}, 
      $g_j$ over $\smash{\calL_{p^{(\ell,1)},p^{(r,1)}}}$ satisfies 
        conditions of $\Tarski^*$ during the for loop and thus,
      both $\smash{p^{(\ell,1)}\preceq p^{(r,1)}}$ and \Cref{induc} hold 
      at the end of the for loop. 
This shows that both of them hold at the end of the first main loop.
      
The induction step is similar. 
Assume that both conditions hold for $p^{(\ell,j)},p^{(r,j)}$ for
  $j\in [i-1]$.      
We start by showing that 
  $p^{(\ell,i)}, p^{(r,i)}$ on line \ref{define pr} satisfy both 
  $p^{(\ell,i)}\preceq p^{(r,i)}$ and \Cref{induc}.
   
To prove $p^{(\ell,i)}\preceq p^{(r,i)}$, we  
make the following observation: 
  $p^{(\ell,j_1)}\preceq p^{(r,j_2)}$
   for all $j_1,j_2\in [i-1]$ with $q^{j_1}\preceq q^{j_2}$.
We divide the proof into three cases:
\begin{flushleft}\begin{enumerate}
\item[] \textbf{Case 0:} $j_1=j_2=j$.
Trivially follows from $p^{(\ell,j)}\preceq p^{(r,j)}$.		
\item[] \textbf{Case 1:} $j_1<j_2$.  By the inductive hypothesis, we know $p^{(\ell,j_2)}\preceq p^{(r,j_2)}$. Considering the while loop $j_2$, by the definition on line \ref{define pr} and $q^{j_1}\preceq q^{j_2}$, we know $p^{(\ell,j_1)}\preceq p^{(\ell,j_2)}$ before the loop on line \ref{for loop}. Furthermore, by the updating rule on line \ref{find pl pr} , we know that $p^{(\ell,j_2)}$ is monotonically non-decreasing, which means $p^{(\ell,j_1)}\preceq p^{(\ell,j_2)}$ holds after while loop $j_2$. So we can derive that after while loop $j_2$, $p^{(\ell,j_1)}\preceq p^{(r,j_2)}$.
		
		\item[]\textbf{Case 2:} $j_1>j_2$.  By the inductive hypothesis, we know $p^{(\ell,j_1)}\preceq p^{(r,j_1)}$. Considering the while loop $j_1$, by the definition on line \ref{define pr} and $q^{j_1}\preceq q^{j_2}$, we know $p^{(r,j_1)}\preceq p^{(r,j_2)}$ before the loop on line \ref{for loop}. Furthermore, by the updating rule on line \ref{find pl pr}, we know that $p^{(r,j_1)}$ is monotonically non-increasing, which means $p^{(r,j_1)}\preceq p^{(r,j_2)}$ holds after while loop $j_1$. So we can derive that after while loop $j_1$, $p^{(\ell,j_1)}\preceq p^{(r,j_2)}$.
		\end{enumerate}\end{flushleft}
		Now we move back to our proof of $p^{(\ell,i)}\preceq p^{(r,i)}$. For every $j_1,j_2\in [i-1]$ such that $q^{j_1}\preceq q^i\preceq q^{j_2}$, we know that $p^{(\ell,j_1)}\preceq p^{(r,j_2)}$. So the same partial order relation of the least upper bound of $p^{(\ell,j_1)}$ and the greatest lower bound of $p^{(r,j_2)}$ also holds, i.e., $p^{(\ell,i)}\preceq p^{(r,i)}$ before the loop on line \ref{for loop}. 
		
Next we prove \Cref{induc}	on line \ref{define pr}.	
	For each $t\in[a]$, given that $p^{(\ell,i)}$ is the $\LUB$,
	there must exist $j^*\in[i-1]$ such that 
$$
q^{j^*}\preceq q^i,\quad
	p^{(\ell,j^*)}\preceq p^{(\ell,i)}\quad \text{and}\quad 
	p^{(\ell,j^*)}_t=p^{(\ell,i)}_t.$$ Since $(p^{(\ell,j^*)},q^{j^*})\preceq (p^{(\ell,i)},q^i)$, we have 
$$p^{(\ell,j^*)}_t+g(p^{(\ell,j^*)},q^{j^*})_t\preceq 
p^{(\ell,i)}_t+
g(p^{(\ell,i)},q^i),$$ 
which implies $g(p^{(\ell,j^*)},q^{j^*})_t\leq g(p^{(\ell,i)},q^i)_t$.
On the other hand, we have
  $g(p^{(\ell,j^*)},q^{j^*})_t\ge 0$ by the inductive hypothesis.
So $ g(p^{(\ell,i)},q^{i})_t\ge 0$. 
$ g(p^{(r,i)},q^{i})_t\le 0$ can be proved similarly.

Given that both $p^{(\ell,i)}\preceq p^{(r,i)}$ and \Cref{induc}
  hold at the beginning of the loop on 	
	line \ref{for loop}, the rest of the proof is essentially the same as the proof 
    in the base case.
It is easy to prove by induction that both $p^{(\ell,i)}\preceq p^{(r,i)}$ and \Cref{induc} hold
      at the beginning of each for loop on
      line \ref{for loop}, 
      $g_j$ over $\smash{\calL_{p^{(\ell,i)},p^{(r,i)}}}$ satisfies 
        conditions of $\Tarski^*$ during the for loop and thus,
      both $\smash{p^{(\ell,i)}\preceq p^{(r,i)}}$ and \Cref{induc} hold 
      at the end of the for loop. 
This shows that both of them hold at the end of the main while loop.

		This completes the induction and the proof of the lemma. 
	\end{proof}

	\begin{lemma}\label{lemma: usage of RefinedTarski*}
		At the end of every round $i$, we have
		$g(p_1,q^i)_t=g(p_2,q^i)_t$ for all
		  $p_1,p_2\in\calL_{p^{(\ell,i)},p^{(r,i)}}$ and all
		 $t\in[a+b+1]\setminus [a]$,
	\end{lemma}
	\begin{proof}
		Consider the end of round $t$ of the for loop on line \ref{for loop}. If $g(p^{(\ell,*)},q^i)_{t}=1$, then for every $p^{(\ell,*)}\preceq p\preceq p^{(r,*)}$, we have $1=g(p^{(\ell,*)},q^i)_{t}\leq g(p,q^i)_{t}$, i.e., $g(p,q^i)_{t}=1$. Similarly, if $g(p^{(r,*)},q^i)_{t}=-1$, then we~have 
		$g(p,q^i)_{t}=-1$
		for every $p^{(\ell,*)}\preceq p\preceq p^{(r,*)}$. For the last case of $g(p^{(\ell,*)},q^i)_{t}= g(p^{(r,*)},q^i)_{t}=0$, for every $p^{(\ell,*)}\preceq p\preceq p^{(r,*)}$, we have $0=g(p^{(\ell,*)},q^i)_{t}\leq g(p,q^i)_{t}\leq g(p^{(r,*)},q^i)_{t}=0$, i.e., $g(p,q^i)_{t}=0$.

		For subsequent round $t+1,t+2,\cdots$ of the for loop, we know that $\calL_{p^{(\ell,i)},p^{(r,i)}}$ can only shrink. 
		So the property remains. This finishes the proof of the lemma.
	\end{proof}
	
We are now ready to prove the two lemmas needed for the correctness of Algorithm  \ref{alg: solve Tarski*}:
	
	\begin{lemma}\label{lemma: rmonotonicity}
		For any two rounds $j_1$ and $j_2$, if $q^{j_1}\preceq q^{j_2}$, then $(q^{j_1},0)+r^{j_1}\preceq (q^{j_2},0)+r^{j_2}$.
	\end{lemma}
	\begin{proof}
	We consider two cases when $j_1<j_2$ and when $j_1>j_2$.
	
		\textbf{Case 1:} $j_1<j_2.$ Considering the round $j_2$, by the definition on line \ref{define pr} and $q^{j_1}\preceq q^{j_2}$, we know that $p^{(\ell,j_1)}\preceq p^{(\ell,j_2)}$ before the loop on line \ref{for loop}. 
		In the loop on line \ref{for loop}, when $p^{(\ell,j_2)}$ is updated by $p^{(\ell,*)}$,  $p^{(\ell,j_2)}$ is monotonically non-decreasing.
  So at the end of the loop we still have $p^{(\ell,j_1)}\preceq p^{(\ell,j_2)}$.
  Given that $(p^{(\ell,j_1)},q^{j_1})\preceq (p^{(\ell,j_2)},q^{j_2})$, we have for every $t\in [a+b]\setminus [a]$:
		\begin{equation*}
		 q^{j_1}_{t-a}+r^{j_1}_{t-a}
		   =q^{j_1}_{t-a}+g(p^{(\ell,j_1)},q^{j_1})_t\leq 
		   q^{j_2}_{t-a}+g(p^{(\ell,j_2)},q^{j_2})_t=
		   q^{j_2}_{t-a}+r^{j_2}_{t-a},
		\end{equation*}
        and $r_{b+1}^{j_1}\leq r_{b+1}^{j_2}$.
		
		\textbf{Case 2:} $j_1>j_2.$ Case 2 is analogous to Case 1. Considering the round $j_1$, by the definition on line \ref{define pr} and $q^{j_1}\preceq q^{j_2}$, we know that $p^{(r,j_1)}\preceq p^{(r,j_2)}$ before the loop on line \ref{for loop}. 
		In the loop on line \ref{for loop}, when $p^{(r,j_1)}$ is updated by $p^{(r,*)}$, $p^{(r,j_1)}$ is monotonically non-increasing, so that $\forall t\in[a+b+1]\setminus [a]$, $g(p^{(r,j_1)},q^{j_1})_t$ is non-decreasing. So 
		we have $p^{(r,j_1)}\preceq p^{(r,j_2)}$ at the end of the loop on line \ref{for loop}.
		Using $(p^{(r,j_1)},q^{j_1})\preceq (p^{(r,j_2)},q^{j_2})$ and \Cref{lemma: usage of RefinedTarski*},
		we have for every $t\in [a+b]\setminus [a]$:
		\begin{align*}
		  q^{j_2}_{t-a}+r^{j_2}_{t-a}
		    &=q^{j_2}_{t-a}+g(p^{(\ell,j_2)},q^{j_2})_t
		    = q^{j_2}_{t-a}+g(p^{(r,j_2)},q^{j_2})_t\\  
		    &\geq q^{j_1}_{t-a}+g(p^{(r,j_1)},q^{j_1})_t
		    =q^{j_1}_{t-a}+g(p^{(\ell,j_1)},q^{j_1})_t
		    =q^{j_1}_{t-a}+r^{j_1}_{t-a},
		\end{align*}
		and $r_{b+1}^{j_2}\geq r_{b+1}^{j_1}$.
		
		This completes the proof of the lemma. 
	\end{proof}
	\begin{lemma}\label{lemma: rsolution}
		At the end of each round $i$, if $r_t^i\geq 0$ for $t\in [b+1]$, 
			then $(p^{(\ell,i)},q^i)$ is a solution to $\textsc{Tarski}^*(n,a+b)$ on $g$;
		if $r_t^i\leq 0$ for $t\in[b+1]$, then $\smash{(p^{(r,i)},q^i)}$ is a solution to $\smash{\textsc{Tarski}^*(n,a+b)}$ on $g$.
	\end{lemma}
	\begin{proof}
	By \Cref{hehehehe} we have $g(p^{(\ell,i)},q^i)_t\geq 0$ and $g(p^{(r,i)},q^i)_t\leq 0$ for all $t\in[a]$. So if $r_t^i\geq 0$ for all $t\in [b+1]$, then $\smash{g(p^{(\ell,i)},q^i)_t=r_{t-a}^i\ge 0}$ for all $t\in [a+b+1]\setminus [a]$
	and thus, $\smash{(p^{(\ell,i)},q^i)}$ is a solution to $\textsc{Tarski}^*(n,a+b)$ on $g$. Similarly if $r_t^i\leq 0$ for all $t\in [b+1]$, then $\smash{(p^{(r,i)},q^i)}$ is a solution to $\textsc{Tarski}^*(n,a+b)$ on $g$.
	\end{proof}

	\textbf{Query complexity.} For each round of Algorithm \ref{alg: solve Tarski*},  Algorithm \ref{alg: solve RefinedTarski*} is called $b+1$ times on line \ref{find pl pr} and each call of Algorithm \ref{alg: solve RefinedTarski*} will use $O(q(n,a))$ queries. The outer algorithm $\calB$ has no more than $q(n,b)$ rounds,
	which means the query complexity of Algorithm \ref{alg: solve Tarski*} is $O((b+1)\cdot q(n,a) \cdot q(n,b))$.
\end{proof}	

\section{Discussion and Open Problems}
While progress  has been made on improving the upper bounds for finding Tarski fixed points, the techniques for lower bounds remain limited.

For the black-box (query complexity) model studied in this paper, the key question left open is to close the  gap between $\Omega(\log^2 n)$ and $\smash{O(\log^{\lceil (k+1)/2\rceil} n)}$. The first gap is from $\Tarski(n,4)$, where the lower bound is $\Omega(\log^2 n)$ and the upper bound is $O(\log^3 n)$. Note that if one could improve the lower bound of $\Tarski(n,4)$ to get a tight bound $\Theta(\log^3 n)$, it would imply that $O(\log^2 n)$ is tight for $\Tarski^*(n,3)$ (while the tight bounds of $\Tarski^*(n,1)$ and $\Tarski^*(n,2)$ are $\Theta(\log n)$). Or even relaxing the goal, is it possible to prove a lower bound $\Omega(\log^3 n)$ for $\Tarski(n,k)$ when $k$ is a constant, say, $k=100$?

With regards to the white-box model, it is known that $\Tarski$ is in the intersection of PPAD and PLS~\cite{etessami2020tarski}, and so is in CLS~\cite{fearnley2021complexity} and EOPL~\cite{goos2022further}. It would also be very interesting to see if $\Tarski$ is complete for some computational complexity classes.

\section*{Acknowledgements}
We thank anonymous reviewers for helpful comments on an earlier draft.

\begin{flushleft}
\bibliographystyle{alpha}

\newcommand{\etalchar}[1]{$^{#1}$}

\end{flushleft}
\end{document}